\documentclass[11pt,a4paper]{article}

\usepackage[T1]{fontenc}
\usepackage[T1,hyphens]{url}
\usepackage{amsfonts, amsmath, amssymb}  
\usepackage{float}
\usepackage{parskip}
\usepackage{graphicx}
\usepackage{soul}
\usepackage{orcidlink}

\newtheorem{theorem}{Theorem}

\newtheorem{remark}{Remark}
\newtheorem{assumption}{Assumption}

\newtheorem{objection}{Objection}
\newenvironment{proof}{\textsf{\bfseries{}Proof. }}{\hfill$\square$\par}

\usepackage{verbatim}
\usepackage{enumerate}

\usepackage{physics}

\usepackage{times}
\usepackage{sectsty}
\usepackage{xcolor,colortbl}
\usepackage[framemethod=tikz]{mdframed}
\usepackage{caption}
\usepackage{adjustbox}
\usepackage{float}
\usepackage{systeme}
\usepackage{multicol}
\usepackage{graphicx}
\usepackage[margin=1in]{geometry}
\usepackage{url}
\usepackage{hyperref}

\def\({\left(}
\def\){\right)}

\newcommand{\hilbert}{\mathcal{H}}
\newcommand{\mc}[1]{\mathcal{#1}}

\newcommand{\pobs}[1]{#1}
\newcommand{\obs}[1]{\mathsf{\pobs{#1}}}
\newcommand{\uop}[1]{\widehat{\mathbf{#1}}}

\newcommand{\iprod}[2]{\left\langle#1\middle|#2\right\rangle}
\newcommand{\oprod}[2]{\left|#1\middle\rangle\middle\langle#2\right|}

\makeatletter
\@ifundefined{widetext}{}{}
\@ifundefined{turnpage}{}{}
\@ifundefined{ruledtabular}{}{}
\@ifundefined{acknowledgments}{%
}{}
\makeatother

\providecommand{\pacs}[1]{}
\providecommand{\preprint}[1]{}

\begin{document}

\title{ND-Photonic QRNGs in a Noisy Environment}

%\author{J. M. Ag\"{u}ero Trejo, C. S. Calude, %O. C. Stoica}
\author{J. M. Ag\"{u}ero Trejo\footnote{School of Computer Science, University of Auckland, New Zealand. Email: \href{mailto:jagu688@aucklanduni.ac.nz}{jagu688@aucklanduni.ac.nz}.} 
%\href{mailto:manuel.aguero15@gmail.com}.{manuel.aguero15@gmail.com}.}
\ \orcidlink{0000-0001-9631-0326}, 
Cristian S. Calude\footnote{Corresponding author. School of Computer Science, University of Auckland, New Zealand. Email: \href{mailto:c.calude@auckland.ac.nz}{c.calude@auckland.ac.nz}.}% \href{mailto:cscalude@gmail.com}{cscalude@gmail.com}.}
\ \orcidlink{0000-0002-8711-6799}, 
O. C. Stoica\footnote{Department of Theoretical Physics, NIPNE--HH, Bucharest, Romania. 	Email: \href{mailto:cristi.stoica@theory.nipne.ro}{cristi.stoica@theory.nipne.ro}.}
%\href{mailto:holotronix@gmail.com}{holotronix@gmail.com.}.} 
\ \orcidlink{0000-0002-2765-1562}}

\date{}

\maketitle
\begin{abstract}
Standard pseudo-random generators have weaknesses that have led to the development of quantum random number generators (QRNGs). However, common QRNG validation methods, whether based on quantum indeterminism or statistical tests, are insufficient to guarantee high-quality randomness. In contrast, a mathematical theory based on the Located Kochen-Specker Theorem proves that 3D QRNGs produce maximally unpredictable outputs without using entanglement, and both theory and experiments have supported their security. The paper focuses on a practical photonic implementation of a 3D QRNG that is easier to deploy than cryogenic superconducting implementations. As any physical implementation is subject to various measurement errors, it is important to study theoretically and experimentally the type and role of errors in 3D QRNGs. In this paper, we will model the photonic 3D QRNG as an open quantum system, constructed as an arrangement of imperfect beam-splitters with a range of losses based on the fidelity of its components, and we will show that under certain conditions, the process remains within the scope of the Kochen-Specker Theorem, which guarantees maximum unpredictability.
\end{abstract}


%%%%%%%%%%%%%%%%%%%%%%%%%%%%%%%%%%%%%%
\section{Introduction}
Random numbers play an essential role in many applications, such as cryptography.  Currently, almost all practical applications rely on Pseudo-Random Number Generators (PRNGs), which have well-known pitfalls (predictability, short periods, and non-uniform distribution).
This problem has motivated the development of several Quantum Random Number Generators (QRNGs), the best known being Quantis~\cite{Quantis2020Q,mannalath_comprehensive_nodate}.  To ensure a QRNG is suitable for security applications, either its outputs are postulated to be unpredictable based on the ``intrinsic'' indeterminism of quantum mechanics, and/or its unpredictability is tested with statistical tests.  Both arguments are too weak to certify the quality of randomness.

3D QRNGs that generate provably maximally unpredictable strings by measuring localized value indefinite variables have been designed and studied~\cite{aguero_trejo_new_2021}.  These QRNGs do not use entanglement and are provably better than {\it any} PRNG~\cite{trejo_photonic_2023}.  In addition, robust experimental results probing the maximal unpredictability of quantum random strings generated by a 3D QRNG implemented with a superconducting transmon confirmed the theoretical results~\cite{AGUEROTREJO2024114632}.

The photonic implementation of the 3D QRNG based on the located variant of the Kochen-Specker Theorem~\cite{2015-AnalyticKS}, studied in~\cite{RSPA23}, does not require very low temperatures, so it is better suitable for practical applications than the 3D QRNG implemented with a superconducting transmon coupled to a microwave cavity attached to a dilution cryostat cooled down to the base temperature of $\sim$20 mK~\cite{PhysRevLett.119.240501}.

As any physical implementation is subject to various measurement errors, it is important to study, both theoretically and experimentally, the types and roles of errors in QRNGs.  In what follows, we will study this problem for the photonic implementation of the 3D QRNG  studied in~\cite{RSPA23} based on the located variant of the Kochen-Specker Theorem~\cite{2015-AnalyticKS}.  In this noisy environment, we can no longer assume that the operators used to model the photonic embodiment are unitary, hence we will model the 3D QRNG as an open quantum system, constructed as an arrangement of imperfect beam-splitters with a range of losses based on the fidelity of actual components, and we will show that under certain conditions, the process remains within the scope of the Kochen-Specker Theorem which guarantees maximum unpredictability.

We will also generalize the located variant of the Kochen-Specker Theorem~\cite{2015-AnalyticKS} to unsharp measurements, by applying it to the pointer states.

%%%%%%%%%%%%%%%%%%%%%%%%%%%%%%%%%%%%%%
\section{Notation and Definitions}

A $\sigma$-algebra on a set $\mathcal{X}$ is a non-empty collection of subsets of $\mathcal{X}$ closed under complement, countable unions, and countable intersections.  The \emph{Borel space} associated to a topological space $\mathcal{X}$ is the pair $(\mathcal{X},\mathcal{B})$, where $\mathcal{B}$ is the $\sigma$-algebra of Borel sets of $\mathcal{X}$. A \emph{Borel function} is a measurable function between two Borel spaces, \cite{Holevo2011ProbabilisticAndStatisticalAspectsOfQuantumTheory}.

Let $\mathcal{H}$ be a Hilbert space and $\mathcal{B(\mathcal{H})}$ the algebra of bounded operators from $\mathcal{H}$ into itself.
By $\mathcal{H}_N$ we denote an $N$-dimensional Hilbert space and $\mathcal{I}_N$ is the identity map on $\mathcal{B}(\mathcal{H})$.
A quantum map $\xi$ is a mapping from a space of density operators to another space of density operators $\rho \mapsto \rho' = \xi[\rho]$.  A  quantum map $\xi$ is \emph{positive} if it maps positive density operators $\rho$ to positive density operators $\rho'$ and  it is \emph{completely positive} if any trivial extension $$\xi\otimes\mathcal{I}_N = \mathcal{B}(\mathcal{H}_d)\otimes \mathcal{B}(\mathcal{H}_{N})\to \mathcal{B}(\mathcal{H}_d)\otimes \mathcal{B}(\mathcal{H}_{N}) $$ is \emph{positive}.

%%%%%%%%%%%%%%%%%%%%%%%%%%%%%%%%%%%%%%
\section{Background}
\label{s:background}
We denote the observable projecting onto the linear subspace spanned by a vector $\ket{\psi}$ as $\obs{P}_\psi=\frac{\oprod{\psi}{\psi}}{\iprod{\psi}{\psi}}$.
We then fix a positive integer $n > 2$ and let $\mathcal{O} \subseteq \{  \obs{P}_{\psi} \mid   \ket{\psi} \in \mathbb{C}^{n} \}$ be a non-empty set of one-dimensional projection observables on the Hilbert space $\mathbb{C}^{n}$.

A set $\mathcal{C} \subset \mathcal{O}$ is a {\it context} of $\mathcal{O}$ if $\mathcal{C}$ has $n$ elements and for all $\obs{P}_{\psi}, \obs{P}_{\phi} \in \mathcal{C}$ with $\obs{P}_{\psi} \neq \obs{P}_{\phi}, \braket{\psi}{\phi}=$0. 
A  {\it value assignment function} (on $\mathcal{O}$) is a partial function $v:\mathcal{O}\to \{0,1\}$ assigning values to some (possibly all) observables in $\mathcal{O}$.\footnote{The partiality of the  function $v$ means that $v(\obs{P})$ can be $0,1$ or indefinite.}
An observable $\obs{P} \in \mathcal{O}$ is  {\it value definite} (under the assignment function $v$)  if $v(\obs{P})$ is defined, i.e.~$v(P)=0$ or $v(P)=1$; otherwise, it is  {\it value indefinite} (under $v$). Similarly, we call $\mathcal{O}$ {\it value definite} (under $v$) if every observable $\obs{P} \in \mathcal{O}$ is value definite.

\if01
In what follows, we assume the following premises:
\begin{itemize}
\item \textbf{Admissibility.} This assumption guarantees agreement with quantum mechanics predictions.  Fix a set $\mathcal{O}$ of one-dimensional projection observables in $\mathbb{C}^n$ and the value assignment function $v:\mathcal{O}\to \{0,1\}$. Then $v$ is admissible if the following two conditions hold for every context $\mathcal{C}$ : a) if there exists $P\in\mathcal{C}$ with $v(P)=1$, then $v(P')=0$, for all $P'\in \mathcal{C} \setminus \{P\}$, b) if there exists $P\in\mathcal{C}$ with $v(P)=0$, then $v(P')=1$, for all $P'\in \mathcal{C} \setminus \{P\}$.

\item \textbf{Non-contextuality of definite values.} Every outcome obtained by measuring a value definite observable is non-contextual, i.e. it does not depend on other compatible observables which may be measured alongside it.
\item \textbf{Eigenstate principle.} If a quantum system is prepared in the state $\ket{\psi}$, then
$v(\obs{P}_\psi)=1$.
\end{itemize}

\begin{theorem}[Located Kochen-Specker Theorem~\cite{abbott2012strongrandomness,vi-aeverywhere-2014,2015-AnalyticKS}]
\label{thm:LKS}
Assume a quantum system prepared in the state $\ket{\psi}$ in a Hilbert space $\mathbb{C}^n$ with $n\geq 3$ , and
let $\ket{\phi}$ be any quantum state such that $0<\abs{ \braket{\psi}{\phi}}<1$.
Let $\mathcal{O}$ be a set of one-dimensional projection observables on $\mathbb{C}^n$ containing $\obs{P}_{\psi}$ and $\obs{P}_{\phi}$, and  $v:\mathcal{O}\to\{0,1\}$ a value asignment function.
If the above three
conditions are satisfied: i) admissibility, ii) non-contextuality and iii) eigenstate
principle, then the projection observable $\obs{P}_{\phi}$ is value indefinite.
\end{theorem}
\fi
%begin_marius

In what follows, we assume the following premises:
\begin{itemize}
\item \textbf{Admissibility.} This assumption guarantees agreement with quantum mechanics predictions.  Fix a set $\mathcal{O}$ of one-dimensional projection observables in $\mathbb{C}^n$ and the value assignment function $v:\mathcal{O}\to \{0,1\}$. Then $v$ is admissible if the following two conditions hold for every context $\mathcal{C} \subseteq \mathcal{O}$ : a) if there exists $P\in\mathcal{C}$ with $v(P)=1$, then $v(P')=0$, for all $P'\in \mathcal{C} \setminus \{P\}$, b) if there exists $P\in\mathcal{C}$ with $v(P')=0$ for all $P'\in \mathcal{C} \setminus \{P\}$, then $v(P)=1$.

\item \textbf{Non-contextuality of definite values.} Every outcome obtained by measuring a value definite observable is non-contextual, i.e. it does not depend on other compatible observables which may be measured alongside it. (Note: this is captured in the definition of a value assignment function and the assumption that the measurement of a value definite observable yields for it the pre-assigned value.)
\item \textbf{Eigenstate principle.} If a quantum system is prepared in the state $\ket{\psi}$, then
$v(\obs{P}_\psi)=1$.
\end{itemize}

\begin{theorem}[Located Kochen-Specker Theorem~\cite{abbott2012strongrandomness,vi-aeverywhere-2014,2015-AnalyticKS}]
\label{thm:LKS}
Assume a quantum system prepared in the state $\ket{\psi}$ in a Hilbert space $\mathbb{C}^n$ with $n\geq 3$ , and
let $\ket{\phi}$ be any quantum state such that $0<\abs{ \braket{\psi}{\phi}}<1$.
There exists $\mathcal{O}$ a set of one-dimensional projection observables on $\mathbb{C}^n$ containing $\obs{P}_{\psi}$ and $\obs{P}_{\phi}$, so that for every value assignment function $v:\mathcal{O}\to\{0,1\}$ 
If the above three
conditions are satisfied: i) admissibility, ii) non-contextuality and iii) eigenstate
principle, then the projection observable $\obs{P}_{\phi}$ is value indefinite (under $v$).
\end{theorem}
%\end_marius
A measurement is {\it sharp} if it is perfectly accurate. Many important results in the study of quantum phenomena require sharp measurements. Nonetheless, experimentally, it is impossible to perform a perfectly precise measurement; working with \emph{unsharp} measurements can significantly complicate and sometimes invalidate certain results.

Although an unsharp version of the Kochen-Specker exists~\cite{Breuer2002KochenSpeckerTheoremForUnsharpSpin1Observables,Breuer2002KochenSpeckerTheoremForFinitePrecisionSpinOneMeasurements,Breuer2003AnotherNoGoTheoremForHiddenVariableModelsOfInaccurateSpin1Measurements,cabello2002FinitePrecisionMeasurementDoesNotNullifyTheKochenSpeckerTheorem},  so far, no unsharp equivalent of the localised variant of the Kochen-Specker Theorem has been formulated. %Under certain conditions, we can still ensure that the Kochen-Specker Theorem applies.

%\cristi{
A major source of error in integrated photonics is photon loss, whether from escaping waveguides, from detector physical limitations, or from using a source that is not guaranteed to emit exactly one photon at a time. A way to model loss is to include additional modes in the model, such as waveguides that guide the lost photons to outputs not monitored by photon detectors. This implies an increase in the dimension of the Hilbert space, for example, by adding another waveguide to each of the three waveguides to model loss. 
This may be useful in some situations, but in fact, the problem takes care of itself%\hl{ 
if we can guarantee that the source is single-photon: if the photon is lost, the detector will simply not detect a photon, and if two or more detectors click, we are dealing with an external photon, and we ignore the result. In practice, we can ignore undetected photons without affecting the quality of the QRNG.
%by skipping a finite number of digits in the resulting string 
We will use heralded photons: in this approach, the photons are generated in pairs, and the detection of the photon is verified by the simultaneous detection of its ``twin''. If only the twin photon is detected, we are dealing with a loss, which is automatically ignored anyway. If the twin photon is not detected, we will ignore the result, as it may be due to a stray photon.
%}

We will make tests to ensure that the physically realized device is not significantly affected by losses or external noise. Such a test consists of coupling in series another device identical to this one but mirroring its behavior, so that it inverts the unitary transformation implemented by this one. The details are given in the article~\cite{CDMTCS587}.

Another point of caution is that the noise and unsharpness may lead to violations of the inequality $0<\abs{ \braket{\psi}{\phi}}<1$. This would make Theorem~\ref{thm:LKS} inapplicable. The physical meaning of this inequality is that, for each of the detectors, there is always a non-vanishing probability that it detects the output photon. If one of them, corresponding to the observable $\obs{P}_{\phi}$, has zero probability to click, then we get $\abs{ \braket{\psi}{\phi}}=0$. If the other detectors have zero probability to click, $\abs{ \braket{\psi}{\phi}}=1$. Experiments should confirm that the inequality is consistently satisfied. But this is automatically ensured by the tests showing that the theoretical probabilities are obtained.

In this article, we study the behaviour of a QRNG~\cite{CDMTCS583} which selects a value of an indefinite observable in $\mathbb{C}^n$ and then measures it unsharply.

To this end, we will analyse each step of the QRNG.

%%%%%%%%%%%%%%%%%%%%%%%%%%%%%%%%%%%%%%
\section{First and Second Measurement Operator}
\label{s:FirstSecondMeasurement}
%%%%%%%%%%%%%%%%%%%%%%%%%%%%%%%%%%%%%%
\subsection{First measurement operator}
Choose an $N$-dimensional Hermitian operator with non-degenerate spectrum (that is, the operator has no linearly independent eigenvectors for the same eigenvalue).  From Theorem~\ref{thm:LKS}, it follows that for any diagonalisable observable $\obs{O}$ with spectral decomposition $\obs{O}=\sum_{i=1}^N\lambda_i \obs{P}_{\lambda_i}$, where $\lambda_i$ denotes each distinct eigenvalue with corresponding eigenstate $\ket{\lambda_i}$, $\obs{O}$ has a predetermined measurement outcome if and only if each projector in its spectral decomposition has a predetermined measurement outcome. 
    
So ideally, we want an operator with eigenvalues corresponding to the standard Cartesian basis on dimension $N$; hence, the basis states correspond to the $n$ inputs of the final measurement device (more generally, an operator satisfying the requirements up to a change of basis).  We will use the first measurement operator to provide a value definite state (preparation state) so that, together with the second operator, it satisfies Theorem~\ref{thm:LKS}; that is, the eigenstates of the second measurement operator are neither orthogonal nor parallel to the preparation state.

Note that choosing a non-degenerate operator allows us to map each one of its distinct eigenvectors to each input port of a physical arrangement.  Also, in \cite{trejo_photonic_2023}, choosing the $\obs{S}_z = \obs{S}(0,0)$ operator conforms to the Cartesian standard basis and plays an important role when introducing imperfections.

%%%%%%%%%%%%%%%%%%%%%%%%%%%%%%%%%%%%%%
\subsection{Second measurement operator} 
Choose an $N$-dimensional Hermitian and unitary operator with distinct eigenvectors.  This operator must be different from the first measurement operator.

To do the measurement of an observable represented by a Hermitian operator using an integrated photonic implementation, we apply a unitary transformation that maps the eigenbasis of the second measurement operator into the standard Cartesian basis implemented by the detectors of the $N$ modes.

As an example for three dimensions, in~\cite{trejo_photonic_2023}, the first measurement was the generalised spin-1 observable
\begin{equation}
\label{eq:spin-one-gen}
\obs{S}(\theta,\varphi)=
\begin{pmatrix}
\cos{\theta} & \frac{e^{-i\varphi}\sin{\theta}}{\sqrt{2}}& 0\\
\frac{e^{i\varphi}\sin{\theta}}{\sqrt{2}} & 0 & \frac{e^{-i\varphi}\sin{\theta}}{\sqrt{2}}\\
0 & \frac{e^{i\varphi}\sin{\theta}}{\sqrt{2}} & -\cos{\theta}
\end{pmatrix}.
\end{equation}

The $\obs{S}_x= \obs{S}(\frac{\pi}{2},0)$ 
operator was chosen to be the unitary second measurement operator $\obs{U}_x$:

\begin{equation}
\label{eq:second-measurement-operator}
\obs{U}_x =\frac{1}{2}
\begin{pmatrix}
1 & \sqrt{2} & 1\\
\sqrt{2} & 0 & -\sqrt{2}\\
1 & -\sqrt{2} & 1
\end{pmatrix}.
\end{equation}

Note that the degeneracy of $\obs{U}_x$ is not a problem.  To see this, it is useful to analyze the role of these operators and the non-contextuality assumption in the original formulation of the Kochen-Specker Theorem.  First, we recall the following conditions for a value assignment map $v$.

\begin{itemize}
    \item (Self-adjoint operator) For any self-adjoint operator $\obs{A}$ corresponding to an observable $\pobs{A}$, we have that $v(\obs{A})\in\{{o_i}\}$, where $\{o_i\}$ are the eigenvalues of $\obs{A}$. That is, each observable corresponds to an element of physical reality, and the values assigned correspond to the set of possible outcomes.
    \item (Quasi-linearity) The commuting operators $\obs{A},\obs{B}$, that is $[\obs{A},\obs{B}]=0$,  satisfy $v(a\obs{A} + b\obs{B})=aV(\obs{A})+bV(\obs{B})$, where $a,b\in\mathbb{R}$.
    \item (Non-contextuality) All observables are assigned values simultaneously regardless of the measurement context.
\end{itemize}

From Quasi-linearity, it follows that the value map must preserve the algebraic structures of the operators.  That is, for any Borel function $f$, we have that $v(\obs{A}) = f(v(\obs{B}))$, whenever $\obs{A}=f(\obs{B})$. This is the core element leading to a contradiction in many proofs of the Kochen-Specker Theorem.  The contradiction only occurs for degenerate operators in the original formulation as a result of the following properties: If an operator $\obs{A}$ is degenerate, then for some non-degenerate operators $\obs{B},\obs{C}$ and Borel functions $f,g$, we have that 
\begin{equation}
\label{eq:AB-fg}
\obs{A}=f(\obs{B})\text{ and } \obs{A} = g(\obs{C}), \text{ with }[\obs{B},\obs{C}]\neq 0.
\end{equation}

Thus, from Quasi-linearity it follows that
\begin{equation}
\label{eq:fvB-gvC}
v(\obs{A}) = f\big(v(\obs{B})\big)=g\big(v(\obs{C})\big).
\end{equation}

As the sum of the projectors of a complete orthonormal set of states yields the identity operator and orthogonal projectors commute, the sum of their assigned values must be one.  For one-dimensional degenerate operators, we obtain a single projector with value 1 and $N-1$ zero-valued projectors in an $N$-dimensional Hilbert space.
This leads to a contradiction when considering all complete sets of projectors, since we consider a projector to be a function of different non-commuting non-degenerate operators; in other words, degenerate one-dimensional operators must be assigned the same value regardless of which non-degenerate operator it is considered to be a function of.  As a consequence, we are forced to accept the existence of value-indefinite observables or introduce some form of contextuality.

Note that the first measurement operator helps us prepare the measurement context with a corresponding value definite preparation state.  Moreover note that all our Hermitian operators are self-adjoint (this property is also essential for our error handling framework), and degeneracy of an observable can then be simply understood as having more than one measurement basis (or measurement context) for an observable, and is not detrimental when localizing a value-indefinite observable (nor is it required in the localized variant of the Kochen-Specker Theorem).  As Heisenberg said in~\cite{Heisenberg1958-HEIPAP},

\begin{quote}
\emph{``What we observe is not nature itself, but nature exposed to our method of questioning''}
\end{quote}

Taking the 3D-QRNG as an example, note that the first measurement operator prepares an eigenstate of a spin projector along a different axis from the second measurement operator, while satisfying Theorem~\ref{thm:LKS}.  So, since eigenvalue relations are preserved for rotated states \cite{Breuer2002KochenSpeckerTheoremForUnsharpSpin1Observables} (see section \ref{s:breuer-unsharp} for more details), all eigenstates of the sharp spin projectors used in this implementation are eigenstates of the unsharp spin projectors.
Moreover, from Theorem~\ref{thm:LKS} we see that for any diagonalisable observable $\obs{O}$ with spectral decomposition $\obs{O}=\sum^n_{i=1}\lambda_i\obs{P}_{\lambda_i}$ where $\lambda_i$ denotes each distinct eigenvalue with corresponding eigenstate $\ket{\lambda_i}$, $\obs{O}$ has a predetermined measurement outcome if and only if each projector in its spectral decomposition has a predetermined measurement outcome.  Thus, in general, we have the following result.

\begin{theorem}
If the following conditions hold for a QRNG certified via value-indefiniteness:
\begin{enumerate}[a)]
    \item The preparation state is an eigenstate of the first measurement operator.
    \item The second measurement operator does not commute with the first measurement operator.
    \item The choice of preparation state ($\ket{\psi}$) satisfies the conditions of Theorem~\ref{thm:LKS} with respect to each eigenstate ($\ket{\phi}$) of the second measurement operator.%cristi
\end{enumerate}
 Then, unsharp measurements do not affect the value indefiniteness of the outcomes.
 \end{theorem}

Note that these results hold for the binary protocol proposed in \cite{calude_binary_2024} due to its unitary equivalence to the ternary protocol from \cite{trejo_photonic_2023}.  The guarantee of value indefiniteness alone does not ensure that the distribution of outcomes is desirable.  Nonetheless, the conditions above represent a natural arrangement that still allows for a good distribution of outcomes while facilitating a feasible physical implementation

For example in the ternary case, choosing the eigenstates $\ket{\pm 1}$ along the z-axis as preparation states lead to the distribution $\frac{1}{4},\frac{1}{2},\frac{1}{4}$. In the binary case we may choose the state $\ket{1}$ or $\ket{0}$ to obtain the distribution $\frac{1}{2},\frac{1}{2}$ with the measurement operator 

\begin{equation}
\label{eq:measurement-operator-U-prime}
\obs{U}'=\begin{pmatrix}
    \frac{1}{\sqrt{2}} & \frac{1}{\sqrt{2}} & 0\\
    \frac{1}{\sqrt{2}} & \frac{1}{\sqrt{2}} & -1 \\
    0 & 0 & 0
\end{pmatrix}.
\end{equation}

We have seen that, under certain conditions, value indefiniteness is unaffected by unsharp measurements.  Nonetheless, it is important to examine the effects of experimental imperfections on the distribution of outcomes.

%%%%%%%%%%%%%%%%%%%%%%%%%%%%%%%%%%%%%%
\section{Limits of Measurements}
\label{s:limits}
Both imperfect technology and the laws of nature limit the practical realization of ideal quantum measurements. These limitations make our measurements unsharp, but, as we will see, they also provide a way to ensure the value indefiniteness even for unsharp measurements. The reason is that these limitations have the same fundamental origin as the Kochen-Specker Theorem.

Wigner \cite{wigner1952MessungQMOperatoren,Wigner1952MessungQMOperatorenPBusch2010EnTranslation} showed that observables that can be realized as ideal measurements are rare, since they have to satisfy severe constraints due to the conservation laws. His proof concerns the pre-measurement, which is the unitary evolution of the measuring device interacting with the observed system, leading to the superposition of more possible outcomes, before the resolution achieved by invoking the Projection Postulate. Given that this is a unitary process, Wigner used the conservation laws to examine whether they constrain the realizability of accurate, repeatable measurements of sharp observables. It turns out that {\it sharp measurements cannot be both accurate and repeatable}.
His result was generalized by Araki and Yanase~\cite{ArakiYanase1960MeasurementofQMOperators}, which is known as the Wigner-Araki-Yanase (WAY) Theorem \cite{LoveridgeBusch2011MeasurementofQMOperatorsRevisited}.

In this context, accuracy is the requirement that the result of a measurement represents faithfully the value of the measured observable. Repeatability is the requirement that if a measurement is repeated, it yields the same result. For measurements used in preparation, repeatability is simply the accuracy of the preparation. A system prepared so that an observable $\obs{A}$ has a certain value $a$ satisfies repeatability if, by measuring the observable $\obs{A}$, the same value $a$ is obtained. Repeatability directly relates to the Eigenstate principle from Section~\S\ref{s:background}.
For discrete observables, it is possible, in principle, to satisfy either accuracy or repeatability, but not both~\cite{LoveridgeBusch2011MeasurementofQMOperatorsRevisited}.

But Ozawa showed that measurements of non-discrete observables are irrepeatable even if we give up accuracy \cite{Ozawa1984QuantumMeasuringProcessesOfContinuousObservables} (moreover, the limitation of repeatability, but also of accuracy, extends to discrete observables, if the Hamiltonian is non-negative~\cite{Stoica2025CanWeAccuratelyReadWriteQuantumData}).

In particular, Ozawa's result implies that there is no way to choose the axis of a spin measurement with perfect accuracy.
To measure a spin along a particular axis, one needs to align the Stern-Gerlach device's magnets along that axis.
This requires tuning two real parameters, for example, the polar angle $\theta$ and the azimuthal angle $\varphi$.
This would be impossible already in a classical world, because it requires controlling an infinite amount of information.
But in a quantum world, there is an additional kind of obstruction.
An intrinsic quantum source of unsharpness, in addition to noise and experimental errors, is the physical law itself.
Ozawa's theorem on measurements of continuous or non-discrete observables says more than the classical fact that infinite precision is impossible in practice: it implies that, whatever axis one chooses, if one wants to verify how the axis is oriented, one will get a different result. 
In particular, since the measurement of the spin of an atom is based on its interaction with the magnetic field, it will not reflect with perfect accuracy the original set-up of the axis along which one intends to measure the spin.

Under these conditions, one may wonder whether it is still possible to ensure in practice that the conditions of the Located Kochen-Specker Theorem~\ref{thm:LKS} can be satisfied.
In the following Section~\S\ref{s:breuer-unsharp} we will present previous results in this direction.

%%%%%%%%%%%%%%%%%%%%%%%%%%%%%%%%%%%%%%
\section{Breuer's Unsharp Kochen-Specker Theorem}
\label{s:breuer-unsharp}
In this section, we discuss Breuer's generalizations of the Kochen-Specker Theorem to unsharp spin $1$ observables, notice their limits of applicability, and show that they are insufficient to certify the QRNG if the observable is unsharp.

Breuer presented his generalization of the Kochen-Specker Theorem to unsharp spin-$1$ observables in \cite{Breuer2002KochenSpeckerTheoremForFinitePrecisionSpinOneMeasurements,Breuer2002KochenSpeckerTheoremForUnsharpSpin1Observables}. He labeled the states by intermediate values using some thresholds $\delta$ (which he called \emph{unsharpness tolerance}) and $1-\delta$, where $0.5>\delta\geq 0$. Triads are then labeled with the values ``almost true'' (for eigenvalues $\geq 1-\delta$) and ``almost false'' (for eigenvalues $\leq \delta$), reducing the proof of the unsharp case to that of the sharp case. But the results of an experiment are always sharp values, even if the measurement is unsharp, because we read them from the pointer, which is a macroscopic object that displays sharp values. Macroscopically distinct states are orthogonal. So one does not observe the intermediate values assigned by the unsharp commuting observables; they remain ``hidden''. That is, the ``unsharp eigenvalues'' may be indefinite, but they remain unobservable. Also, there is a non-vanishing probability to get the wrong result reported, hence the use of the term ``almost'' in Breuer's articles~\cite{Breuer2002KochenSpeckerTheoremForFinitePrecisionSpinOneMeasurements,Breuer2002KochenSpeckerTheoremForUnsharpSpin1Observables}.

Before stating the Breuer Theorem, we need some notation. We denote the unsharp spin-1 observables, forming a positive operator-valued measure, by $F^{\mathbf{n},\epsilon}(i)$, $i\in\{-1,0,1\}$.
The directions $\mathbf{m}$ are distributed with a density $w_{\mathbf{N},\epsilon}\(\mathbf{m}\)$ around the direction $\mathbf{N}$ in space.
The eigenvalues of the three positive operators $F^{\mathbf{z},\epsilon}(i)$ are three values from the set $\{\alpha_1,\alpha_2,\alpha_3,\alpha_4\}$.
Breuer also assumed that the unsharpness is the result of the experimenter's ignorance of what spin observable, assumed to be sharp, is in fact measured,
\begin{assumption}[Breuer's Unsharpness Assumption]
\label{assumption:breuer}
We can model unsharp measurements as sharp measurements of inaccurately known sharp observables.
\end{assumption}

That is, the assumption is that the spin-1 measurement is sharp along \emph{some} well-defined axis; we don't know with perfect accuracy which axis it is.
This was the assumption of Breuer Theorems~\ref{thm:U-KS-breuer} and \ref{thm:U-KS-breuer-second}.
For example, in ~\cite{Breuer2002KochenSpeckerTheoremForUnsharpSpin1Observables}, p. 197, right before Proposition 2, it is stated and justified as follows:
\begin{quote}
[If the unsharp spin properties] %$F^{n,\varepsilon}(i)$
transform covariantly under rotations, they are angular momentum properties and can be regarded as spin properties with the same justification as the sharp spin properties $\obs{P}_{n,i}$. This is in line with Weyl's idea of defining observables by their transformation properties under some kinematic group.
\end{quote}

It can be formulated equivalently as the assumption that unsharp measurements are \emph{post-processing} of projective measurements (PVMs), which implies that value indefiniteness is achieved. We can use this to build a genuinely random QRNG.

Assumption~\ref{assumption:breuer} is stated in Breuer's papers before he stated his theorem, but it is not restated in the theorem. However, the proof is based on it.
That is, the POVM is assumed to commute, which allows simultaneous diagonalization. This is a strong restriction.
We will discuss this below.

\begin{theorem}[Breuer's Unsharp Kochen-Specker Theorem, v. 1]
\label{thm:U-KS-breuer}
For any unsharpness tolerance $0.5>\delta\geq 0$, if in an unsharp spin-$1$ measurement
\begin{enumerate}
	\item[(A1)] 
the densities of apparatus misalignments transform covariantly, \textit{i.e.} if $w_{R\mathbf{n},\epsilon}\(R\mathbf{m}\)=w_{\mathbf{n},\epsilon}\(\mathbf{m}\)$ for all rotations $R$, and
	\item[(A2)] 
the measurement inaccuracy described by the densities $w_{\mathbf{n},\epsilon}\(\mathbf{m}\)$ of the apparatus misalignments is so small that the unsharp spin observables $F^{\mathbf{n},\epsilon}(i)$ have, for all $i\in\{-1,0,1\}$, the eigenvalues $\alpha_1$ and $\alpha_4$ larger or equal to $1-\delta$, and the eigenvalues $\alpha_2$ and $\alpha_3$ smaller or equal to $\delta$,
\end{enumerate}
Then not all sharp spin observables in a Kochen-Specker set of directions can consistently be assigned approximate truth values in a non-contextual way.
\end{theorem}

As said, in addition to these assumptions, Breuer used Assumption~\ref{assumption:breuer} in his proof, discussed below.

In a subsequent article~\cite{Breuer2003AnotherNoGoTheoremForHiddenVariableModelsOfInaccurateSpin1Measurements}, Breuer avoids approximate truth values.
Instead, he uses four numbers as labels: the four eigenvalues $\alpha_j$ shared by the three unsharp operators.
In this version of his theorem, Breuer makes the very mild assumption that some eigenvalues are distinct. His newer theorem holds under weaker assumptions about measurement inaccuracy, but it still relies on Assumption~\ref{assumption:breuer}.

\begin{theorem}[Breuer's Unsharp Kochen-Specker Theorem, v. 2]
\label{thm:U-KS-breuer-second}
If the densities of apparatus misalignments transform
covariantly, $w_{R\mathbf{n},\epsilon}\(R\mathbf{m}\)=w_{\mathbf{n},\epsilon}\(\mathbf{m}\)$, and if the density of misalignments is such that both $\alpha_1$ and $\alpha_4$ are different from both $\alpha_2$ and $\alpha_3$, then there is a finite set of intended measurement directions for which not all results of inaccurate measurements can be predetermined in a non-contextual way.
\end{theorem}

We will first try to justify Breuer's assumption, but then we will find it too restrictive.

Quantum physicists who studied measuring devices realized that all that can be observed in a Stern-Gerlach experiment is, in fact, the place where the Silver atom hits the metallic plate. This is the pointer state. We observe the pointer state and, from this, infer something about the state of the Silver atom. While the impact mark may be very small, we will call it ``macroscopic'' or ``classical'' because its observation does not require a quantum measuring device. Of course, it is a quantum state, not a classical one, but it behaves quite classically. Anyway, let us use the terms ``macroscopic'' or ``macrostate'' to refer to pointer states.

\begin{remark}
\label{rem:macro-distinct}
Macroscopically distinguishable quantum states are orthogonal.
\end{remark}

This entails that the pointer observables are always sharp.

Another observation, based on Ozawa Theorem about measurements of continuous or non-discrete observables discussed in Section~\S\ref{s:limits}, is the following:
\begin{remark}
\label{rem:spin-axis-inaccuracy}
There is always an inaccuracy in the direction of the axis along which we intend to measure the spin.
\end{remark}

Note that this also applies to photonic implementations or other implementations that require phase shifters to select the observable, since the phases are continuous parameters.

On the other hand, it was argued that Ozawa Theorem can be evaded by relative observables, for example, relative position~\cite{LoveridgeBusch2011MeasurementofQMOperatorsRevisited,Loveridge2020ARelationalPerspectiveOnTheWignerArakiYanaseTheorem} or, in our case, relative orientation in space.

Together, Remarks~\ref{rem:macro-distinct} and \ref{rem:spin-axis-inaccuracy} seem to justify Breuer's Assumption~\ref{assumption:breuer}.

Alas, as noticed by Breuer himself, Assumption~\ref{assumption:breuer} can be satisfied only if the POVM is commutative, that is, any two of its constitutive positive operators commute. Unsharp measurements whose statistics are realizable as smeared  (or fuzzy) sharp measurements are exactly the commutative POVMs~~\cite{Beneduci2016POVMandFellerMarkovKernels_commutative_POVM_IS_smeared_PVM,BuschLahtiPellonpaaYlinen2016QuantumMeasurement}.
This implies that 

\begin{remark}
\label{rem:assumption:breuer}
We cannot rely on Assumption~\ref{assumption:breuer} for unsharp measurements.
Most unsharp measurements are irreducible to smeared sharp measurements.
In most cases, there is no underlying sharp observable being fuzzified.
\end{remark}

\begin{remark}
\label{rem:assumption:breuer-mixtures}
Moreover, even when the POVM is commutative, and hence it gives the same probabilities as a sharp measurement of an unknown observable, it does not imply that the unsharp measurement \emph{really is} a sharp measurement of an unknown observable.
The reason is that there are more ways to interpret the same POVM as a probabilistic distribution over definite states, as the distinction between \emph{improper} and \emph{proper mixtures} shows (\cite{dEspagnat1976ConceptualFoundationsOfQuantumMechanics}, chapter 7.2, \cite{dEspagnat2006OnPhysicsAndPhilosophy}, chapter 8, \cite{dEspagnat2018VeiledRealityAnAnalysisOfPresentDayQuantumMechanicalConcepts}, chapter 7).
This makes it difficult to connect what was measured to what the pointer indicates, which is necessary to determine whether the pointer's state is due to value indefiniteness.
\end{remark}

In addition, other reasons prevent the use of the unsharp Kochen-Specker Theorems proved by Breuer to guarantee value indefiniteness.
In Breuer First Theorem~\ref{thm:U-KS-breuer}, even if the state is an eigenvector of $F^{\mathbf{n},\epsilon}(i)$, this does not guarantee the outcome. The ``almost true'' and ``almost false'' labeling is unobservable, so the value indefiniteness of the labeling has no determinate implication for the resulting outcome. In particular, it cannot be used to guarantee the quality of the numbers produced by a QRNG.
In the second theorem given by Breuer, the state vectors are labeled as functions of the outcome and of the eigenvalues they have under the corresponding unsharp operator. But again, being an eigenvector of $F^{\mathbf{n},\epsilon}(i)$ does not guarantee the outcome. Even under Assumption~\ref{assumption:breuer}, the outcome cannot tell us the actual sharp observable that was measured, and so we cannot tell us what eigenstate and what eigenvalue were really observed. Also, the $4$ possible eigenvalues are shared among $3$ operators.
So even if they were pre-labeled, would this determine the outcome? If the label is $\alpha_1$, $\alpha_2$, or $\alpha_3$, we can't know if the outcome is 1 or -1. Only if it is $\alpha_4$ can we know that the outcome is $0$. Therefore, even if Breuer Theorems are an important generalization of the Kochen-Specker Theorem, it seems that we cannot use them to certify randomness. 

On the other hand, unsharp observables that cannot be understood as sharp measurements of unknown sharp observables already exhibit non-classicality, which can be understood as contextuality (Spekkens' ``generalized contextuality'')~\cite{Spekkens2005ContextualityForPreparationsTransformationsAndUnsharpMeasurements,Kunjwal2019BeyondTheCabelloSeveriniWinterFrameworkMakingSenseOfContextualityWithoutSharpnessOfMeasurements}.
This is precisely because the unsharp observable includes non-commuting positive operators.
Interestingly, this understanding can reveal contextuality in single unsharp measurements and even in two-dimensional Hilbert spaces, for example, by treating the qubit SIC-POVM (symmetric, informationally complete, positive operator-valued measure) as an unsharp observable.
Therefore, a deeper analysis may show that even in the unsharp case, the obtained pointer states arise from value indefiniteness, thereby complementing Breuer's results.

%%%%%%%%%%%%%%%%%%%%%%%%%%%%%%%%%%%%%%%%%%%%%%%%%%%%
\section{The Physicality of the Stinespring Dilation}
\label{s:dilation-physical}
In this article, we will use the Stinespring dilation to show that even if the measurement used by the QRNG is unsharp, it can be understood as a sharp measurement in a larger Hilbert space. This may seem like a mathematical trick in which we invoke a larger Hilbert space conveniently imagined for our purposes. Still, in this section, we will show that the larger Hilbert space has a solid physical basis. The Stinespring dilation is not in an imagined Hilbert space but is already realized concretely and physically in the experimental setup. The usual way we understand measurements as being about the observed system obfuscates this fact, restricting the larger Hilbert space and leading to unsharpness, and weakening the connection between the value-indefinite properties and what the pointer indicates.

To this end, we will recall Ozawa's analysis of measurements and see that it concerns the observer-plus-pointer system as a whole.
Ozawa's analysis reveals that the observable we intend to measure may be different from what we really measure, but even if it is unsharp, the combined observable, the one actually taking place, is sharp. It only seems unsharp if we ignore the larger system and trace out over its degrees of freedom, as is usually done when discussing quantum measurements. Therefore, it is physically justified to use the Stinespring dilation, as it actually occurs in the total system.
We include an extension of the unsharp Located Kochen-Specker Theorem, which is more general than Breuer's, in that it doesn't assume only the restrictive kind of unsharpness, and it ensures that the observed values are indefinite, just by applying the sharp Located Kochen-Specker Theorem to the total system.

Let us start by stating the obvious objection that, by using the Stinespring dilation, we dilate in a fictitious extension of the Hilbert space of the physical systems under consideration.
\begin{objection}
\label{obj:stinespring-abstract}
The construction of the dilation is theoretical, and to be of help in practice, it must be realized physically, because it makes no sense to invoke nonexistent contexts.
\end{objection}

Fortunately, the Stinespring dilation is not only justified but, in fact, how things happen in reality. The most realistic description of a measuring device is based on the Stinespring dilation. In the following, we will explain this. 
We will show that real-world measurements automatically realize the Stinespring dilation, which is known as the \emph{measurement dilation}. 

Let us remember the question ``what does a Stern-Gerlach device measure?''
At first sight, the obvious answer is that it measures the spin of the observed system, which may be, for example, a Silver atom or a neutron. But taking a closer look revealed that it measures the position where the observed Silver atom hits the metallic plate.

Recall that Wigner examined the spin measurement and concluded that such a measurement cannot be both accurate and repeatable (that is, undisturbing) \cite{Wigner1952MessungQMOperatorenPBusch2010EnTranslation}. He then showed that, if the measuring device contained an infinite number of particles, the ideal of a measurement would be attainable in principle. He also proposed other ways to distinguish between the two possible spin values by extending the set of possible outcomes to include an ``error'' outcome and discarding such results. Unfortunately, this idea cannot be used in practice to discard outputs (in particular, outputs of the QRNG) that are caused by noise or experimental error, since this would require knowing which unsharp observable is produced as a result of their effects and building the measuring device accordingly.

But this situation is not at all that bad for us, as we shall see.

Remark~\ref{rem:macro-distinct} is of extreme importance to understand what observable is actually measured in a quantum measurement. We may \emph{intend} to measure a particular observable, and design and build a device to do this job, but ultimately, we have to track back from the pointer state to infer what the actual observable was. This is explained for example in \cite{Ozawa1984QuantumMeasuringProcessesOfContinuousObservables}, \cite{BuschLahtiPellonpaaYlinen2016QuantumMeasurement} Section \S{7.7}, and \cite{Loveridge2020ARelationalPerspectiveOnTheWignerArakiYanaseTheorem} Section \S{2.2}.
We evolve the combined state backward in time, including the measuring device (and, if needed, the environment) and the observed system, and then trace out the measuring device. What remains is the actually measured observable of the system under observation.

In \cite{Ozawa1984QuantumMeasuringProcessesOfContinuousObservables}, Ozawa showed that a measuring device can be constructed, abstractly, for any observable $\pobs{E}$. His construction, called \emph{measurement dilation}, is based on \emph{Na\u{\i}mark Dilation Theorem} \cite{Naimark1940AboutSecondKindSelfAdjointExtensionsOfSymmetricalOperator}, which can be seen as a particular case of \emph{Stinespring Dilation Theorem} \cite{Stinespring1955PositiveFunctionsOnC-star-algebras}.

Let the observed system be $S$, and the measuring device be $A$, with the respective Hilbert spaces $\hilbert_S$ and $\hilbert_A$, and let $\hilbert=\hilbert_S\otimes \hilbert_A$. Let $\obs{Z}\in\mc{L}\(\hilbert_A\)$ be the pointer observable, taken to be a Hermitian operator because we assume it to be sharp, based on Remark~\ref{rem:macro-distinct}.
Let us detail this for the measurement of an $N$-dimensional system, assuming that we measure the observable $\obs{O}=\sum_{i=1}^N\lambda_i \obs{P}_{\lambda_i}$ as in Section~\S{\ref{s:FirstSecondMeasurement}}, where $\ket{\lambda_i}_S$ are the eigenstates of $\obs{O}$ and $\obs{P}_{\lambda_i}=\ket{\lambda_i}_S\bra{\lambda_i}_S$.
Let $\ket{\text{ready}}_A$ and $\ket{\zeta_1}_A,\ldots,\ket{\zeta_N}_A$ be the eigenvectors of the pointer state.
Let $\ket{\psi}_S=\sum_{i=1}^N a_i \ket{\lambda_i}_S$ be the initial state of the observed system.
Then, the pre-measurement consists of entangling the pointer of the measuring device with the observed system, by applying the evolution operator $\obs{U}$:
\begin{equation}
\label{eq:pre-measurement}
\ket{\psi}_S\otimes\ket{\text{ready}}_A\stackrel{\obs{U}}{\mapsto}\sum_{i=1}^N a_i \ket{\lambda_i}_S\otimes\ket{\zeta_i}_A,
\end{equation}
where $a_i=\braket{\lambda_i}{\psi}_S$.

But this is true if the measurement is sharp. In general, this is not the case, and this is precisely the subject of this article.

The pointer is a subsystem of the measuring device, but an observable of a subsystem is also an observable of the full system; we will not need to refine this further in the following. Similarly, on the total system $S+A$, the pointer observable is $\obs{I}_{S}\otimes\obs{Z}$.
The environment, including other systems that may affect our measurements, must also be considered. For simplicity, we will consider it included in the measuring device, since it indeed affects the measurement. Again, there is no need to treat it separately, and this is not even the right way, because we will describe a measurement that is not necessarily ideal.
We assume that the interaction between the measuring device and the observed system begins at time $t_0$, and that the pointer indicates the outcome starting at time $t_1 > t_0$, when the measurement ends.
Let $\uop{U}_{t_1,t_0}$ denote the unitary evolution operator of the combined system during the measurement.

Let $\obs{\varrho}\in\mc{T}\(\hilbert_A\)$ be the ``ready'' state of the measuring device, and $\obs{\rho}\in\mc{T}\(\hilbert_S\)$ the state of the observed system, at $t=t_0$. Let the measured observable be the \emph{positive operator-valued measure} (POVM) $\pobs{E}:\mc{E}\to\mc{L}\(\hilbert_S\)$, where $\mc{E}$ is a $\sigma$-algebra of subsets of a set $\Omega$, representing the space of possible outcomes resulting from the measurement of $\pobs{E}$. The POVM $\pobs{E}$ should satisfy the following \emph{probability reproducibility condition} \cite{Busch2009NoInformationWithoutDisturbance,BuschLahtiPellonpaaYlinen2016QuantumMeasurement}:
\begin{equation}
\label{eq:prob-reprod-cond}
\tr\(\rho\pobs{E}\(X\)\)=\tr\(\uop{U}_{t_1,t_0}\obs{\rho}\otimes\obs{\varrho}\uop{U}_{t_1,t_0}^\dagger\obs{I}_{S}\otimes\pobs{Z}(X)\),
\end{equation}
where $X\in\mc{E}$ and $\pobs{Z}$ is the sharp observable $\obs{Z}(X)$ seen as a POVM.
Let us assume that the ``ready'' state of the pointer is pure, $\obs{\varrho}=\oprod{\phi}{\phi}$, where $\ket{\phi}\in\hilbert_A$. We define the isometric embedding
\begin{equation}
\label{eq:isometric-embedding}
V_{\phi}:\hilbert_S\to\hilbert_S\otimes\hilbert_A, V_{\phi}\ket{\psi}=\ket{\psi}\otimes\ket{\phi},
\end{equation}
for any $\ket{\phi}\in\hilbert_S$.
Then, $\pobs{E}$ can be 
uniquely extracted  from equation \eqref{eq:prob-reprod-cond},
\begin{equation}
\label{eq:measurement-dilation}
\pobs{E}\(X\)=V_{\phi}^\dagger\uop{U}_{t_1,t_0}\(\obs{I}_{S}\otimes\pobs{Z}(X)\)\uop{U}_{t_1,t_0}V_{\phi}.
\end{equation}

This is a particular case of the Stinespring construction for the observable $\pobs{E}$, realized as the pointer observable $\obs{I}_{S}\otimes\pobs{Z}(X)$, evolved back in time to $t_0$. It is called the \emph{measurement dilation}.

\begin{remark}[Practical realizability of sharp measurements]
\label{rem:measurement-dilation-b}
In practice, it is difficult, if not impossible, to build an ideal measuring device for an intended observable $\pobs{E}$, so in fact equation \eqref{eq:measurement-dilation} should be understood rather as a definition of what observable a particular measuring device measures.
\end{remark}

\begin{remark}[Unsharpness of the measured observable]
\label{rem:measurement-unsharpness}
Equation~\eqref{eq:measurement-dilation} implies that, in fact, the measurements are usually unsharp when understood as being about the observed system.
Even though Remark~\ref{rem:macro-distinct} implies that the pointer observable itself is sharp, the pointer state indicates \emph{an unsharp observable $\pobs{E}$ of the observed system} $S$, because of the restriction done by the isometric embedding $V_{\phi}:\hilbert_S\to\hilbert_S\otimes\hilbert_A$ from equation~\eqref{eq:isometric-embedding}.
\end{remark}

\begin{remark}[The physicality of the Stinespring dilation]
\label{rem:measurement-dilation}
At the same time, the isometric embedding $V_{\phi}:\hilbert_S\to\hilbert_S\otimes\hilbert_A$ from equation~\eqref{eq:isometric-embedding} and taking the trace in equation~\eqref{eq:prob-reprod-cond} do not represent physical processes, but they describe the observer's inference about the observed system $S$, and the observer's \emph{ignorance} of the total state. In reality, the pointer observable corresponds to the observable $\uop{U}_{t_1,t_0}\(\obs{I}_{S}\otimes\pobs{Z}(X)\)\uop{U}_{t_1,t_0}$, which, being a unitary transformation of the sharp pointer observable $\obs{I}_{S}\otimes\pobs{Z}(X)$, it is sharp too.
This is important because it shows that the Stinespring dilation is realized in practice in every quantum measurement.
It does not have to be realized in an imagined auxiliary Hilbert space, because it is already realized physically. Still, the standard descriptions of the measurement restrict what the pointer state tells us to the observed system $S$.
\end{remark}

%%%%%%%%%%%%%%%%%%%%%%%%%%%%%%%%%%%%%%%%%%%%%%%%%%%%
\section{The Problem of the Overlap of Unsharp Outcomes}
\label{s:caution-overlap}
Another potential objection to using the Stinespring Dilation Theorem, due to the unsharpness, is that
\begin{objection}
\label{obj:overlap}
The unsharpness of the measurements leads to overlap between the probability distributions of the possible outcomes, and we need to ensure it does not compromise the quality of randomness. 
\end{objection}

This overlap is, in fact, because in practice the measurements are unsharp, but this is true only if we consider the observable to be restricted to the observed system, and not to be an observable of the whole system.

As an extreme example of such an overlap, note that, for an ideal 3D-Photonic QRNG, random numbers are generated. Still, the total unprojected state vector is determined by the unitary evolution and therefore does not exhibit randomness. This means that the randomness of the three alternative outcomes cancels out in the unprojected state, which is what we get if we lose track of which outcome obtains, for example, by using a single detector for all three possible modes. This corresponds to replacing the projective measurement in the 3D-Photonic QRNG with the trivial POVM consisting solely of the identity operator. While this is an extreme example that completely washes out the value of indefiniteness, it illustrates the problem posed by unsharpness.

However, once we understood that the pointer observable is sharp and that the observed system is used to put the pointer state into various mutually orthogonal states, we realized that we do not actually care about the real state of the observed system. Our random numbers are generated by the detectors in response to the photons, not by the photons themselves. %\cristi{
It can be argued that since the detectors have to reset between measurements, this prevents performing the measurements alongside other compatible measurements. Therefore, one cannot speak of the Kochen-Specker Theorem in this case~\cite{GrudkaKurzynski2008IsThereContextualityForASingleQubit}. This is true if we want to test the theorem, for example, for violations of inequalities. But all we need is to ensure value-indefiniteness for a single measurement, guaranteed by counterfactual, not by actual measurements. For this reason, the failure of detectors to reset exactly in the same state does not affect value indefiniteness, because alternative observables needed for this exist in principle in the extended Hilbert space, which is physical, even if they are not \emph{de facto} measured. Another argument against POVM contextuality~\cite {GrudkaKurzynski2008IsThereContextualityForASingleQubit} is that the dilation is not unique. However, in our case, we refer to Ozawa's measurement dilation, which is physical and uniquely given by the pointer observable. So we can focus only on the pointer observable, and apply the Localized Kochen-Specker Theorem to its observation, and not merely to the measurement of the observed system.

In the Ozawa measurement dilation, the observable measured may be unsharp, but the total system undergoes a sharp measurement; see equation~\eqref{eq:measurement-dilation} and its discussion.
 Note that a sharp measurement can still be performed on an impure state, and can still result in an impure state, see, for example, the case of the EPR experiment.
The reason is that a reduced density operator can also be projected, and if the observable is a degenerate operator, the state may remain impure.
Hence, Ozawa's measurement dilation does not guarantee a pure state.
We assume the existence of a larger system that includes the observed system and the pointer, and that the larger system is in a pure state. 
Then, if the observed system was not in the same state as the one resulting from the pre-measurement, $\abs{\braket{\Psi}{\Phi}}<1$. At the same time, since $\ket{\Phi}$ is obtained by a non-trivial projection from $\ket{\Psi}$, the inequality $\abs{\braket{\Psi}{\Phi}}>0$ also holds.

We now list some assumptions that will allow us to generalize Theorem~\ref{thm:LKS}, and explain why they are automatically satisfied.
\begin{enumerate}
	\item 
There exists, physically, a larger system that includes the observed system and the pointer, and whose total state is pure before and after the measurement. This makes sense because in most formulations of quantum mechanics, at least the total state of the universe is assumed to be a pure state. We denote by $\hilbert_A$ the Hilbert space of the pointer, by $\hilbert_S$ the Hilbert space of the observed system and whatever environment needs to be included, and by $\hilbert=\hilbert_S\otimes \hilbert_A$ the total Hilbert space.
	\item
The preparation and pre-measurement result in a total state of the combined system represented by the state vector $\ket{\Psi}$.
	\item Immediately after the measurement, the system can be found with nontrivial probabilities in at least two macroscopically distinct states $\ket{\Phi}$ and $\ket{\Phi'}$, both eigenstates of the extended pointer observable $\obs{I}_{S}\otimes\obs{Z}$.
\end{enumerate}

The last condition shows that a nontrivial projection occurred.
The reason is that, from the above assumptions, $\ket{\Psi}=a\ket{\Phi}+a'\ket{\Phi'}+\ldots$, where $\ket{\Phi},\ket{\Phi'}\ldots$ are mutually orthogonal, being macroscopically distinct, and at least $a=\braket{\Phi}{\Psi}$ and $a'=\braket{\Phi'}{\Psi}$ are not $0$. Then, this implies that $\abs{\braket{\Phi}{\Psi}}<1$, and from
\begin{equation}
\label{eq:Psi_Phi_nontrivial-proof}
1=\abs{\braket{\Psi}{\Psi}}^2=\abs{\braket{\Psi}{\Phi}}^2+\abs{\braket{\Psi}{\Phi'}}^2+\ldots>\abs{\braket{\Psi}{\Phi}}^2
\end{equation}
one obtains
\begin{equation}
\label{eq:Psi_Phi_nontrivial}
0<\abs{\braket{\Psi}{\Phi}}<1.
\end{equation}

For more clarity, let us see how this works in the case of a sharp measurement. For sharp measurements, $\ket{\Psi}=\sum_{i=1}^N a_i \ket{\lambda_i}_S\otimes\ket{\zeta_i}_A$ and $\ket{\Phi}=\ket{\lambda_k}_S\otimes\ket{\zeta_k}_A$ as in equation~{\eqref{eq:pre-measurement}}. Then, a way to ensure $0<\abs{\braket{\Psi}{\Phi}}<1$ is if $a_k=\braket{\lambda_k}{\psi}\neq 0$ and at least another component $a_j\neq 0$, where $j\neq k$. Then, $\ket{\Phi'}=\ket{\lambda_j}_S\otimes\ket{\zeta_j}_A$. This is what happens in the sharp version, Theorem~{\ref{thm:LKS}}. But the same condition \eqref{eq:Psi_Phi_nontrivial} can be ensured regardless of the state of the observed system, because $\braket{\zeta_k}{\zeta_j}_A\neq 0$, which is satisfied because $\ket{\zeta_k}$ and $\ket{\zeta_j}$ represent macroscopically distinct states. In other words, the same value of indefiniteness obtained by applying Theorem~{\ref{thm:LKS}} to the states of the observed system can be obtained by applying Theorem~{\ref{thm:LKS}} to the states of the pointer.
But this condition is satisfied by the pointer states even if the measurement is not sharp, because even in this case the macroscopically distinct pointer states $\ket{\Phi}$ and $\ket{\Phi'}$ are orthogonal, but $\ket{\Psi}$ is a nontrivial linear combination of them. Then, the condition $0<\abs{\braket{\Psi}{\Phi}}<1$ can be validated simply by verifying empirically that the measurement sometimes yields $\ket{\Phi}$, and sometimes yields $\ket{\Phi'}$.

Under these assumptions, we can prove the following result by simply applying Theorem~\ref{thm:LKS} to the pointer states.
\begin{theorem}[An unsharp Located Kochen-Specker Theorem]
\label{thm-LKS-unsharp-general}
For possibly unsharp measurements, assuming that identically prepared systems result in nontrivial probabilities in the macroscopically distinct total pointer observables $\ket{\Phi}$ and $\ket{\Phi'}$, the observable $\obs{P}_{\Phi}$ is value indefinite.
\end{theorem}
\begin{proof}
As we have seen, these conditions ensure that $0<\abs{\braket{\Psi}{\Phi}}<1$.
Then, we can apply Theorem~\ref{thm:LKS} to the pointer states, from which it follows that the total observable $\obs{P}_{\Phi}$ is value indefinite.
\end{proof}

\begin{remark}
\label{rem:thm-LKS-unsharp-general}
Can we extend the conclusion of Theorem~{\ref{thm-LKS-unsharp-general}} to the reduced observable that we measure on the observed system?
We don't know, but we do not need it. The whole point of Theorem~{\ref{thm-LKS-unsharp-general}} was to avoid the necessity to extend value indefiniteness to unsharp observables.
This allows us to evade Objection~\ref{obj:overlap}.
Recall from Section~\S{\ref{s:dilation-physical}} that what we can observe is the pointer observable $\pobs{Z}(X)$ of the measuring device, and not directly an observable of the observed system.
The pointer observable $\pobs{Z}(X)$ is sharp and extends to the sharp observable $\obs{I}_{S}\otimes\pobs{Z}(X)$ on the total system, containing the observed system and the measuring device.
To obtain the measured observable, we apply backward the time evolution and obtain $\uop{U}_{t_1,t_0}\(\obs{I}_{S}\otimes\pobs{Z}(X)\)\uop{U}_{t_1,t_0}$ as in equation~{\eqref{eq:measurement-dilation}}.
To obtain the measured observable $\pobs{E}(X)$ of the system that we intended to observe, we apply to this the isometric embedding $V_{\phi}$ from equation~{\eqref{eq:isometric-embedding}} and its adjoint $V_{\phi}^\dagger$, as in equation~{\eqref{eq:measurement-dilation}}.
This is the inverse operation of the dilation.
If we are interested in the observed system, what we measure in fact is the observable $\pobs{E}(X)$, regardless of what observable we intended to measure.
But what we measure sharply is only the pointer observable $\obs{I}_{S}\otimes\pobs{Z}(X)$.
And this is perfect for our analysis of the unsharp measurements, since all measurements, sharp or unsharp, are modeled by the Ozawa measurement dilation described in Section~\S{\ref{s:dilation-physical}}.
To use the results of the measurement as random numbers, all we need is to ensure the value indefiniteness of the observable $\obs{P}_{\Phi}$, where $\ket{\Phi}$ is an eigenvector of $\obs{I}_{S}\otimes\pobs{Z}(X)$, and not of the observable $\pobs{E}(X)$. This is what we achieve in Theorem~{\ref{thm-LKS-unsharp-general}}. Whether or not this implies the value indefiniteness of the observable $\pobs{E}(X)$ is irrelevant, as long as we have shown that the projector associated with the pointer state is value indefinite. Hence, its measurement is a random digit.
\end{remark}

%%%%%%%%%%%%%%%%%%%%%%%%%%%%%%%%%%%%%%
\section{The Effect of Experimental Imperfection in the Distribution of Outcomes}

We illustrate the effect of experimental imperfection in the distribution of outcomes, and the solution based on the Stinespring dilation, for a 3D QRNG.

One of the main ways to implement beamsplitters in integrated photonics is to use \emph{multimode interferometers} (MMIs). Here, single-mode waveguides are coupled to a multimode fiber with a specified number of allowed modes. 
Moreover, phase modulation can be achieved by using a thermal phase shifter. 
For example, a thermal shifter may be implemented in silicon devices by connecting a resistive silicon strip adjacent to the waveguide to a metal pad, where a voltage is applied to control the temperature.

Thus, by integrating two MMIs (acting as balanced beamsplitters) with a thermal phase shifter, we may construct a \emph{Mach-Zehnder Interferometer} (MZI).

This solution offers two main advantages when implementing our device. 
\begin{enumerate}
    \item MMI's performance is close to that of an ideal balanced beamsplitter.
    \item A MZI characterized by a tunable phase-shifter can be mapped onto a beam splitter with tunable reflectivity. 
\end{enumerate}

Thus, the first advantage allows us to obtain the preparation state described by the first measurement operator, and, together with the second advantage, we can design a physically realizable arrangement of Mach-Zehnder interferometers that implements our desired unitary operator.

Thus, MZIs are a suitable element to model errors in a photonic implementation. An MZI is composed of two balanced beamsplitters and two phase shifters (one internal and one external),
\begin{equation}
\label{eq:MZI}
\begin{aligned}
\obs{M} &= 
\begin{pmatrix}\frac{1}{\sqrt{2}} & \frac{i}{\sqrt{2}} \\\frac{i}{\sqrt{2}} & \frac{1}{\sqrt{2}} \\ \end{pmatrix}
\begin{pmatrix}e^{i\phi_2} & 0 \\ 0 & 1  \\ \end{pmatrix}
\begin{pmatrix}\frac{1}{\sqrt{2}} & \frac{i}{\sqrt{2}} \\\frac{i}{\sqrt{2}} & \frac{1}{\sqrt{2}} \\ \end{pmatrix}
\begin{pmatrix}e^{i\phi_1} & 0 \\ 0 & 1  \\ \end{pmatrix}
\\
&=ie^{\frac{i\phi_2}{2}}
\begin{pmatrix}e^{i\phi_1}\sin{\frac{\phi_2}{2}} & \cos{\frac{\phi_2}{2}} \\
e^{i\phi_1}\cos{\frac{\phi_2}{2}} & -\sin{\frac{\phi_2}{2}} \\ 
\end{pmatrix}
\end{aligned}
\end{equation}

Thus, we may introduce loss by setting different angles to the beam splitters as follows:
\begin{equation}
\label{eq:MZI-loss}
\widetilde{\obs{M}} =
\begin{pmatrix}\cos{\beta} & i\sin{\beta} \\i\sin{\beta} & \cos{\beta} \\ \end{pmatrix}
\begin{pmatrix}e^{i\phi_2} & 0 \\ 0 & 1  \\ \end{pmatrix}
\begin{pmatrix}\cos{\alpha} & i\sin{\alpha} \\i\sin{\alpha} & \cos{\alpha} \\ \end{pmatrix}
\begin{pmatrix}e^{i\phi_1} & 0 \\ 0 & 1  \\ \end{pmatrix}
\end{equation}

Letting
\begin{equation}
\label{eq:fixed-params-xy}
\begin{cases}
x &= \cos{(\alpha -\beta)}\sin{(\frac{\phi_2}{2})} - i\cos{(\alpha +\beta)}\cos{(\frac{\phi_2}{2})} \\
y &= \sin{(\alpha +\beta)}\cos{(\frac{\phi_2}{2})} - i\sin{(\alpha +\beta)}\sin{(\frac{\phi_2}{2})} \\
\end{cases}
\end{equation}
we have

\begin{equation}
\label{eq:MZI-loss-fixed-params}
\widetilde{\obs{M}} = ie^{\frac{i\phi_2}{2}}
\begin{pmatrix} e^{i\phi_1}x & y \\
e^{i\phi_1}y & -x \\ 
\end{pmatrix}
.
\end{equation}

Similarly, the phase values will be imprecise in the real case scenario, so adjusting $\widetilde{\obs{M}}$ to $\widetilde{\obs{M}}(\phi_1',\phi_2')$ and solving for $\widetilde{\obs{M}}(\phi_1',\phi_2') = \obs{M}$ describes a restriction induced by errors. For example, for $\phi_2$ we get
\begin{equation}
\label{eq:phi_2-prime}
\phi_2' = 2\arcsin{\sqrt{
\frac{\sin^2(\frac{\phi_2}{2}) - \cos^2( \alpha + \beta )}
{\cos^2 ( \alpha - \beta ) - \cos^2(\alpha + \beta) }} }.
\end{equation}

That is,
\begin{equation}
\label{eq:ineq}
2\abs{\alpha + \beta - \frac{\pi}{2}} < \phi_2 < \pi - 2\abs{\alpha - \beta}.
\end{equation}

Note that this type of imperfections may be addressed by the use of additional phase shifters to obtain a Hermitian operator $\widetilde{\obs{U}}$ with high fidelity, without affecting the result from the previous sections. Now, there is an additional consideration we must take during the two-dimensional decomposition of our measurement operator.

For example, for the unitary matrix $\obs{U}_x$ from equation~\eqref{eq:second-measurement-operator}, in the 3-dimensional case we have $\obs{U}_x = \obs{B}_{1,2}^{-1}\cdot \obs{B}_{2,3} \cdot \obs{D} \cdot \obs{B}_{1,2}$ where
\begin{equation}
\label{eq:matrices-BD}
\begin{aligned}
\obs{B}_{1,2} &=
\begin{pmatrix}
\sqrt{\frac{2}{3}} & \frac{1}{\sqrt{3}} & 0\\
\frac{i}{\sqrt{3}} & -i\sqrt{\frac{2}{3}} & 0\\
0 & 0 & 1
\end{pmatrix},
&
\obs{B}_{2,3} &=
\begin{pmatrix}
1 & 0 & 0\\
0 & \frac{1}{2} & -\frac{i\sqrt{3}}{2}\\
0 & \frac{i\sqrt{3}}{2} & -\frac{1}{2}
\end{pmatrix}, \\
\obs{B}_{1,2}^{-1} &=
\begin{pmatrix}
\sqrt{\frac{2}{3}} & -\frac{i}{\sqrt{3}} & 0\\
\frac{1}{\sqrt{3}} & i\sqrt{\frac{2}{3}} & 0\\
0 & 0 & 1
\end{pmatrix},&
 \obs{D}&=\begin{pmatrix} 
1 & 0 & 0\\
0 & -1 & 0\\
0 & 0 & -1
\end{pmatrix}.\\
\end{aligned}
\end{equation}

These matrices are generated through Clements' unitary decomposition \cite{ClementsEtAl2016OptimalDesignForUniversalMultiportInterferometers}, \cite{Cilluffo2024ClementsUnitaryDecompositionInPractice}, where off diagonal elements of $\obs{U}_x$ are nulled by applying transformations  $\obs{T}_{m,n}$ by beam-splitters between channels $m$ and $n$ with $m=n-1$ such as the resulting $\obs{B}_{1,2},\obs{B}_{2,3}$ and $\obs{B}_{1,2}^{-1}$ matrices.
Nonetheless, in a practical setting, imperfections in the MZI interferometers used to carry out these transformations make it impossible to realize a perfect cross $(s=0)$ and bar $(s=\infty)$ states where $s$ is the splitting ratio. If $\theta$ is the ideal angle for each beam splitter in a MZI, then with $\theta + \sigma = \alpha$ and $ \theta + \rho = \beta$ we have that the splitting ratio is restricted by $\abs{\sigma + \rho} \leq \abs{s} \leq \cot{\abs{\sigma - \rho}}$ where $\sigma,\rho$ are the error offsets. In the bar state, the inputs connect directly to the outputs, while in the cross state, each input is directed to the opposite output. So, if a given nulling step falls within these ``forbidden regions, nulling is imperfect, and an off-diagonal residual prevents perfect diagonalization of the matrix, leading to an ``uncorrectable'' error. For example, the matrix $\obs{B}_{1,2}$ has a splitting ratio of $s = \obs{T}_{1,1}/\obs{T}_{1,2} = \left.\sqrt{\frac{2}{3}} \middle/ \frac{1}{\sqrt{3}}\right. = \sqrt{2}$ so it does not fall into a ``forbidden'' region. Nonetheless, this restriction on splitting ratios means that other unitarily equivalent choices of operators may offer very desirable properties at the expense of greater sensitivity to errors. 

Let 
\begin{equation}
\label{eq:matrices-AB}
\obs{A} =\frac{1}{\sqrt{2}}
\begin{pmatrix} 
\sqrt{2} & 0 & 0\\
0 & 1 & 1\\
0 & 1 & -1
\end{pmatrix}, \, \obs{B} =\frac{1}{2}
\begin{pmatrix} 
\sqrt{2} & \sqrt{2} & 0\\
1 & -1 & -\sqrt{2}\\
1 & 1 & -\sqrt{2}
\end{pmatrix}.
\end{equation}

Now consider the unitarily equivalent operator $\obs{U}'=\obs{A}\obs{B}$. We may implement this operator by concatenating two $\obs{U}_x$ interferometers and adjusting their angles and phases as in Table~\ref{ABtable}  (also see Figure~\ref{fig:arrangements}). We can use this table to see that its splitting ratios fall into a ``forbidden'' region.

\begin{table}[tbp]
    \centering
    \begin{tabular}{|r|c|c|}
    \hline
     & phase & angle \\
     \hline\hline
    $\obs{A}$ & 0 & $\frac{\pi}{2}$ \\
    \hline
     & $\pi$ & $\frac{\pi}{6}$ \\
     \hline
     & $-\pi$ & $\frac{\pi}{2}$ \\
     \hline\hline
    $\obs{B}$ & $\pi$ & $\frac{\pi}{6}$ \\
    \hline
     & $-\frac{\pi}{2}$ & $\frac{\pi}{2}$ \\
     \hline
     & 0 & $\frac{\pi}{6}$\\
     \hline
    \end{tabular}
    \caption{angles and phases for operators $\obs{A}$ and $\obs{B}$.}
    \label{ABtable}
\end{table}

\begin{figure}[tbp]
    \centering
    \includegraphics[width=\linewidth]{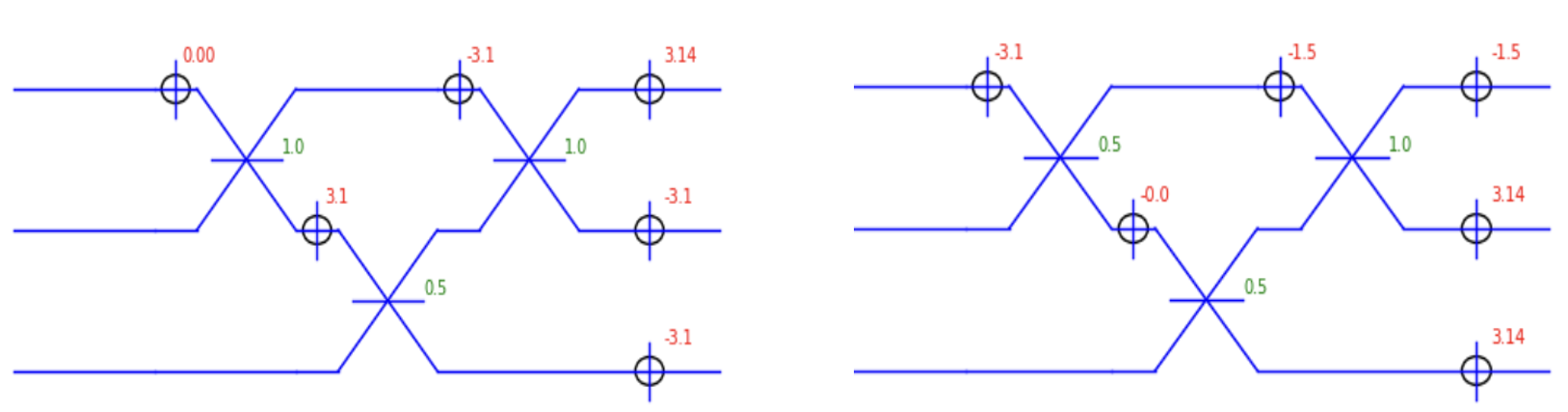}
    \caption{$\obs{A}$ and $\obs{B}$ arrangements.}
    \label{fig:arrangements}
\end{figure}

However, there are ways to mitigate this issue at the cost of additional complexity; for example, we may displace the cross and bar states by adding a splitter at the input, performing the transformation $s \to (s+i\tan \eta) / (1+ is \tan \eta)$. Alternatively, we may find a unitarily equivalent operator with the same output distribution to circumvent this issue.

We are now ready to examine the overall effect of errors on the outcome distribution.

Consider the completely positive map (CP-map)
\begin{equation}
\label{eq:CP-map}
\Phi: L(\mathcal{H})\to L(\mathcal{H}), \varrho:\mathcal{H}\to\mathcal{H},Tr(\varrho)=1,\varrho\geq0
\end{equation}
where $\varrho$ is a positive semi-definite operator acting on the Hilbert space $\mathcal{H}$ such that
\begin{equation}
\label{eq:Kraus}
\Phi(\varrho) = \sum_{i}\obs{K}_i\varrho \obs{K}_i^\dagger \text{ and }\sum_i \obs{K}_i^\dagger \obs{K}_i \leq I.
\end{equation}

Note that for $\varrho^2=\varrho$, we have the projection operator; that is, there exists a $\ket{\psi}\in\mathcal{H}$ such that $\varrho=\ket{\psi}\bra{\psi}$.

We call $\{\obs{K}_i\}$ a Kraus operator which need not be unitary (note that if it is unitary we have $\varrho\mapsto \obs{U}\varrho \obs{U}^\dagger$).

\begin{theorem}[Stinespring dilation \cite{Stinespring1955PositiveFunctionsOnC-star-algebras}]
Let $A$ be a $\mathcal{C}^*$-algebra, $\mathcal{H}$ be a Hilbert space and $B(\mathcal{H})$ be  bounded operator on $\mathcal{H}$. Then, for every CP-map $\Phi : A\to B(\mathcal{H})$ there exist a Hilbert space $\mathcal{K}$ and a unital *-homomorphism $\varphi: A\to B(\mathcal{K})$ such that $\Phi (a) = \obs{V}^* \varphi(a)\obs{V}$, where $\obs{V}:\mathcal{H}\to\mathcal{K}$ is a bounded operator.
\end{theorem}

So a CP-map can be expressed as a map $\obs{V}^* (\cdot)\obs{V}$ in a ``larger space''. Thus, we can see a quantum channel realized by partial tracing a unitary operator acting on a potentially larger Hilbert space. In other words, a CP-map acts as a partial observation of a unitary operation in a larger space where
\begin{equation}
\label{eq:CP-map-action}
\Phi(\varrho) = tr_\mathcal{K}\{\obs{V}(\varrho\otimes \ket{E}\bra{E}) \obs{V}^\dagger\}
\end{equation}
where $\ket{E}$ is an ancilla state.
So, since $\widetilde{\obs{U}}$ is Hermitian, we can describe it as a non-unitary Kraus operator interacting with the environment in a larger space that follows a unitary evolution.

Note that the unitary equivalence of an operator allows us to reduce any operators with the same eigenvalues to the $\obs{U}_x$ case, thus generalizing this result to any unitarily equivalent operator and its corresponding decomposition. In the case of the operator with binary outcomes $\obs{U}'$, this allows us to dismiss altogether the outcomes from the probability zero output port, where the forbidden regions arising from the splitting ratios will make this a non-zero branch by a small margin. By doing so, at the expense of slower bit-generation speed, we may adjust the ``real'' angles and phase values to mitigate bias in the two main output ports while retaining the indefiniteness of the values.

To model the real distribution of outcomes, we can describe the measurement performed by $\widetilde{\obs{U}}$ as a completely positive map acting on a noisy quantum channel 
\begin{equation}
\xi(\rho)=\sum_k \obs{E}_k \rho \obs{E}_k^\dagger,
\end{equation}
where $\rho$ is our density operator and $\obs{E}_k$ are the Kraus operators describing the effect of errors. To broadly model the possible effect of errors, we set $\obs{E}_k$ to represent the depolarising channel on a single qutrit. That is, with some probability, the qutrit remains intact or undergoes a bit flip and/or phase flip error.

Let
\begin{equation}
\label{eq:matrices-A-H}
\begin{array}{c}
\obs{A} = \begin{pmatrix}
    1 & 0 & 0\\
    0 & \omega^2 & 0 \\
    0 & 0 & \omega
\end{pmatrix},\;
\obs{B} = \begin{pmatrix}
    1 & 0 & 0\\
    0 & \omega & 0 \\
    0 & 0 & \omega^2
\end{pmatrix},\\
\obs{C} = \begin{pmatrix}
    0 & 0 & 1\\
    1 & 0 & 0 \\
    0 & 1 & 0
\end{pmatrix}, \; \obs{D}=\obs{C}^\dagger
, \; \obs{E} = \begin{pmatrix}
    0 & \omega & 0\\
    0 & 0 & \omega^2 \\
    1 & 0 & 0
\end{pmatrix}\\
\obs{F} = \begin{pmatrix}
    0 & 0 & \omega^2\\
    1 & 0 & 0 \\
    0 & \omega & 0
\end{pmatrix}, \; 
\obs{G} = \begin{pmatrix}
    0 & \omega^2 & 0\\
    0 & 0 & \omega \\
    1 & 0 & 0
\end{pmatrix} ,\; 
\obs{H} = \begin{pmatrix}
    0 & 0 & \omega\\
    1 & 0 & 0 \\
    0 & \omega^2 & 0
\end{pmatrix}
\end{array}
\end{equation}

where $\omega$ is the third root of unity. Let $p$ be the expected likelihood of errors, estimated by assessing the efficiency of the physical implementation, for example, by conducting the test described in~\cite{CDMTCS587}.
Then, the depolarising channel is given by the Kraus operators
\begin{equation}
\begin{aligned}
 \obs{E}_0 &=\sqrt{1-p}\obs{I}_3, & \obs{E}_1 &=\sqrt{\frac{p}{8}}\obs{D}, & \obs{E}_2 &=\sqrt{\frac{p}{8}}\obs{B},\\
\obs{E}_3 &=\sqrt{\frac{p}{8}}\obs{C}, & \obs{E}_4 &=\sqrt{\frac{p}{8}}\obs{E}, & \obs{E}_5 &=\sqrt{\frac{p}{8}}\obs{F}, \\
\obs{E}_6 &=\sqrt{\frac{p}{8}}\obs{G}, & \obs{E}_7 &=\sqrt{\frac{p}{8}}\obs{H}, & \obs{E}_8 &=\sqrt{\frac{p}{8}}\obs{A}.
\end{aligned}
\end{equation}

So
\begin{equation}
\xi(\rho) = \sum_k \obs{E}_k\rho \obs{E}_k^\dagger = (1-p)\rho + \frac{p}{8}\left(\obs{A}\rho \obs{A}^\dagger + \obs{B}\rho \obs{B}^\dagger + \ldots + \obs{H}\rho \obs{H}^\dagger \right).
\end{equation}

Then, for a qutrit prepared in state $\ket{1}$ we have 

\begin{equation}
(1-p)\ket{1}\bra{1} + \frac{p}{8}\left( 3\ket{-1}\bra{-1} + 2\ket{1}\bra{1} + 3\ket{0}\bra{0} \right).
\end{equation}

Thus, we get the outcome probabilities: 
\begin{equation}
\label{eq:outcome-prob}
\begin{aligned}
\bra{1_x}\xi(\rho)\ket{1_x} &= \frac{1}{4} + \frac{3p}{32}, \\
\bra{0_x}\xi(\rho)\ket{0_x} &= \frac{1}{2} - \frac{3p}{16}, \\
\bra{-1_x}\xi(\rho)\ket{-1_x} &= \frac{1}{4} + \frac{3p}{32}. \\ 
\end{aligned}
\end{equation}

Thus, the conditions for the localized Kochen-Specker Theorem are still met, and we can estimate the outcome probabilities by setting the expected error rate $p$.
We can apply this same method to the binary case and any valid preparation state. That is, for 

\begin{equation}
\label{eq:measurement-operator-U-prime-b}
\obs{U}'=\begin{pmatrix}
    \frac{1}{\sqrt{2}} & \frac{1}{\sqrt{2}} & 0\\
    \frac{1}{\sqrt{2}} & \frac{1}{\sqrt{2}} & -1 \\
    0 & 0 & 0
\end{pmatrix},
\end{equation}
with preparation states $\ket{1}$ or $\ket{0}$,  $\obs{U}\xi(\rho)\obs{U}^\dagger$ yields an outcome in either of the first two outputs ports with probability $\frac{1}{2}\pm \frac{3p}{16}$.

Note that whenever the preparation states belong to the orthonormal basis states, we have that $\xi(\rho)$ is diagonal. Moreover, the unitary decomposition process nullifies the diagonal elements at each step. Thus, by measuring the detectors' efficiency and applying this analysis component-wise during the unitary decomposition, we may manipulate the bias using additional phase shifters in the outputs, in a similar way to how we may displace the cross and bar states. In this manner, one can reduce and distribute the bias as needed to ensure we remain within the localized Kochen-Specker bounds.

%%%%%%%%%%%%%%%%%%%%%%%%%%%%%%%%%%%%%%
\section{Conclusions}
Despite errors, we can ensure that the behavior of this class of QRNGs corresponds to a Hermitian process, which is a completely positive map acting on a quantum channel. The measurement operator is no longer unitary, but it can be modeled as a non-unitary operator that interacts with the environment in a larger unitary system.

Of course, it is essential to ensure the quality of the QRNG by having experimental ways to verify that the implementation is as accurate and as close to unitarity as possible. The symmetry inverter testing device will verify the unitarity of the QRNG \cite{CDMTCS587}. In combination with the verification of the reproducibility of the theoretically predicted probabilities, and the results that the Kochen-Specker Theorem~\ref{thm:LKS} is robust under unsharpness, the viability of the QRNG can be ensured.

In addition, the preparation stage allows the use of any state that is not parallel nor orthogonal to an eigenstate of the
second measurement operator~\cite{2015-AnalyticKS}. This permits flexibility in choosing operators that preserve the value indefiniteness certification under unsharp measurements. Hence, under certain conditions, the process remains within the scope of the Kochen-Specker Theorem, guaranteeing a high degree of incomputability~\cite{acs-2015-info6040773}.

Moreover, to understand the effect of error margins arising from experimental imperfections~\cite{landsman2017foundations}, we examined the measurement performed by an imperfect unitary as a completely positive map acting on a noisy quantum channel. In this manner, we were able to model the effects of the depolarising channel acting on a single qubit. This allows us to develop a partial model of the impact that different sources of errors (e.g., dark counts from the single
photon detectors) have on the distribution of outputs.

\section*{Acknowledgment}
We thank Professors A.~Cabello, S.~Ma and M.~Zimand for their comments, which improved the paper.

% Uncomment and complete if required by the target journal.
% \section*{Declarations}
% \begin{itemize}
% \item Funding: Not applicable.
% \item Conflict of interest/Competing interests: Not applicable.
% \item Ethics approval and consent to participate: Not applicable.
% \item Consent for publication: Not applicable.
% \item Data availability: Not applicable.
% \item Materials availability: Not applicable.
% \item Code availability: Not applicable.
% \item Author contribution: These authors contributed equally to this work.
% \end{itemize}

%% BioMed_Central_Bib_Style_v1.01

%\bibliographystyle{abbrv}
%\bibliography{manu,cris_s,cristi}

\begin{thebibliography}{10}

\bibitem{abbott2012strongrandomness}
A.~A. Abbott, C.~S. Calude, J.~Conder, and K.~Svozil.
\newblock {Strong Kochen-Specker theorem and incomputability of quantum
  randomness}.
\newblock {\em Physical Review A}, 86(062109), Dec 2012.

\bibitem{vi-aeverywhere-2014}
A.~A. Abbott, C.~S. Calude, and K.~Svozil.
\newblock Value indefiniteness is almost everywhere.
\newblock {\em Physical Review A}, 89(3):032109--032116, 2014.

\bibitem{acs-2015-info6040773}
A.~A. Abbott, C.~S. Calude, and K.~Svozil.
\newblock A non-probabilistic model of relativised predictability in physics.
\newblock {\em Information}, 6(4):773--789, 2015.

\bibitem{2015-AnalyticKS}
A.~A. Abbott, C.~S. Calude, and K.~Svozil.
\newblock A variant of the {K}ochen-{S}pecker theorem localising value
  indefiniteness.
\newblock {\em Journal of Mathematical Physics}, 56:102201, Oct 2015.

\bibitem{aguero_trejo_new_2021}
J.~M. Ag{\"u}ero~Trejo and C.~S. Calude.
\newblock A new quantum random number generator certified by value
  indefiniteness.
\newblock {\em Theoretical Computer Science}, 862:3--13, Mar. 2021.

\bibitem{RSPA23}
J.~M. Ag{\"u}ero~Trejo and C.~S. Calude.
\newblock Photonic ternary quantum random number generators.
\newblock {\em Proc. R. Soc. A}, 479:1--16, 2023.

\bibitem{CDMTCS583}
J.~M. {Ag\"{u}ero Trejo} and C.~S. Calude.
\newblock An {$N$}--dimensional quantum random number generator.
\newblock Report {CDMTCS}-583, Centre for Discrete Mathematics and Theoretical
  Computer Science, University of Auckland, Auckland, New Zealand, May 2025.

\bibitem{trejo_photonic_2023}
J.~M. Agüero~Trejo and C.~S. Calude.
\newblock Photonic ternary quantum random number generators.
\newblock {\em Proceedings of the Royal Society A: Mathematical, Physical and
  Engineering Sciences}, 479(2273):20220543, May 2023.

\bibitem{AGUEROTREJO2024114632}
J.~M. {Agüero Trejo}, C.~S. Calude, M.~J. Dinneen, A.~Fedorov, A.~Kulikov,
  R.~Navarathna, and K.~Svozil.
\newblock How real is incomputability in physics?
\newblock {\em Theoretical Computer Science}, 1003:114632, 2024.

\bibitem{CDMTCS587}
J.~M. {Agüero Trejo}, C.~S. Calude, and O.~C. Stoica.
\newblock Testing the 3{D} {QRNG} by undoing.
\newblock Report {CDMTCS}-587, Centre for Discrete Mathematics and Theoretical
  Computer Science, University of Auckland, Auckland, New Zealand, Dec. 2025.

\bibitem{ArakiYanase1960MeasurementofQMOperators}
H.~Araki and M.~Yanase.
\newblock Measurement of quantum mechanical operators.
\newblock {\em Phys. Rev.}, 120(2):622, 1960.

\bibitem{Beneduci2016POVMandFellerMarkovKernels_commutative_POVM_IS_smeared_PVM}
R.~Beneduci.
\newblock Positive operator valued measures and feller markov kernels.
\newblock {\em J. Math. Anal. Appl.}, 442(1):50--71, 2016.

\bibitem{Breuer2002KochenSpeckerTheoremForFinitePrecisionSpinOneMeasurements}
T.~Breuer.
\newblock {K}ochen-{S}pecker theorem for finite precision spin-one
  measurements.
\newblock {\em Physical review letters}, 88(24):240402, 2002.

\bibitem{Breuer2002KochenSpeckerTheoremForUnsharpSpin1Observables}
T.~Breuer.
\newblock A {K}ochen-{S}pecker theorem for unsharp spin 1 observables.
\newblock In {\em Non-locality and Modality}, pages 195--203. Springer,
  Heidelberg, 2002.

\bibitem{Breuer2003AnotherNoGoTheoremForHiddenVariableModelsOfInaccurateSpin1Measurements}
T.~Breuer.
\newblock Another no-go theorem for hidden variable models of inaccurate spin 1
  measurements.
\newblock {\em Philosophy of Science}, 70(5):1368--1379, 2003.

\bibitem{Busch2009NoInformationWithoutDisturbance}
P.~Busch.
\newblock ``{N}o information without disturbance'': {Q}uantum limitations of
  measurement.
\newblock In W.~Myrvold, J.~Christian, and P.~Pearle, editors, {\em Quantum
  Reality, Relativistic Causality, and Closing the Epistemic Circle}, volume~73
  of {\em Western Ontario Series in Philosophy of Science}, pages 229--256,
  Heidelberg, 2009. Springer.

\bibitem{Wigner1952MessungQMOperatorenPBusch2010EnTranslation}
P.~Busch.
\newblock {Translation of ``Die Messung quantenmechanischer Operatoren'' by EP
  Wigner}.
\newblock {\em Arxiv preprint quant-ph/1012.4372}, pages 1--10, December 2010.
\newblock \href{http://arxiv.org/abs/1012.4372}{arXiv:quant-ph/1012.4372}.

\bibitem{BuschLahtiPellonpaaYlinen2016QuantumMeasurement}
P.~Busch, P.~Lahti, J.-P. Pellonp{\"a}{\"a}, and K.~Ylinen.
\newblock {\em Quantum measurement}, volume~23.
\newblock Springer, Switzerland, 2016.

\bibitem{cabello2002FinitePrecisionMeasurementDoesNotNullifyTheKochenSpeckerTheorem}
A.~Cabello.
\newblock Finite-precision measurement does not nullify the {K}ochen-{S}pecker
  theorem.
\newblock {\em Phys. Rev. A}, 65(5):052101, 2002.

\bibitem{calude_binary_2024}
C.~S. Calude and K.~Svozil.
\newblock Binary quantum random number generator based on value indefinite
  observables.
\newblock {\em Scientific Reports}, 14(1):12845, June 2024.

\bibitem{Cilluffo2024ClementsUnitaryDecompositionInPractice}
D.~Cilluffo.
\newblock Commentary on the decomposition of universal multiport
  interferometers: how it works in practice.
\newblock {\em Arxiv preprint quant-ph/2412.11955}, pages 1--13, 2024.
\newblock \href{http://arxiv.org/abs/2412.11955}{arXiv:quant-ph/2412.11955}.

\bibitem{ClementsEtAl2016OptimalDesignForUniversalMultiportInterferometers}
W.~R. Clements, P.~C. Humphreys, B.~J. Metcalf, W.~S. Kolthammer, and I.~A.
  Walmsley.
\newblock Optimal design for universal multiport interferometers.
\newblock {\em Optica}, 3(12):1460--1465, 2016.

\bibitem{dEspagnat2018VeiledRealityAnAnalysisOfPresentDayQuantumMechanicalConcepts}
B.~d'Espagnat.
\newblock {\em Veiled reality: {A}n analysis of present-day quantum mechanical
  concepts}.
\newblock CRC Press, Boca Raton, FL, USA, 2018.

\bibitem{dEspagnat1976ConceptualFoundationsOfQuantumMechanics}
B.~d’Espagnat.
\newblock {\em Conceptual Foundations of {Q}uantum {M}echanics}.
\newblock W.A. Benjamin, New York, 1976.

\bibitem{dEspagnat2006OnPhysicsAndPhilosophy}
B.~d’Espagnat.
\newblock {\em On physics and philosophy}.
\newblock Princeton University Press, Princeton, NJ, 2006.

\bibitem{GrudkaKurzynski2008IsThereContextualityForASingleQubit}
A.~Grudka and P.~Kurzy{\'n}ski.
\newblock Is there contextuality for a single qubit?
\newblock {\em Phys. Rev. Lett.}, 100(16):160401, 2008.

\bibitem{Heisenberg1958-HEIPAP}
W.~Heisenberg.
\newblock {\em Physics and Philosophy: {T}he Revolution in Modern Science}.
\newblock Harper, New York, 1958.

\bibitem{Holevo2011ProbabilisticAndStatisticalAspectsOfQuantumTheory}
A.~Holevo.
\newblock {\em Probabilistic and statistical aspects of quantum theory},
  volume~1 of {\em PSNS}.
\newblock Springer Basel, Basel, Switzerland, 2011.

\bibitem{Quantis2020Q}
{ID Quantique SA}.
\newblock {\em Random Number Generation -- White Paper. What is the Q in QRNG?}
\newblock idQuantique, Geneva, Switzerland, May 2020.

\bibitem{PhysRevLett.119.240501}
A.~Kulikov, M.~Jerger, A.~Poto{\v{c}}nik, A.~Wallraff, and A.~Fedorov.
\newblock Realization of a quantum random generator certified with the
  {Kochen-Specker } theorem.
\newblock {\em Phys. Rev. Lett.}, 119:240501, Dec 2017.

\bibitem{Kunjwal2019BeyondTheCabelloSeveriniWinterFrameworkMakingSenseOfContextualityWithoutSharpnessOfMeasurements}
R.~Kunjwal.
\newblock Beyond the {C}abello-{S}everini-{W}inter framework: {M}aking sense of
  contextuality without sharpness of measurements.
\newblock {\em Quantum}, 3:184, 2019.

\bibitem{landsman2017foundations}
K.~Landsman.
\newblock {\em Foundations of Quantum Theory: From Classical Concepts to
  Operator Algebras}, volume 188 of {\em Fundamental Theories of Physics}.
\newblock Springer, Cham, Switzerland, 2017.

\bibitem{Loveridge2020ARelationalPerspectiveOnTheWignerArakiYanaseTheorem}
L.~Loveridge.
\newblock A relational perspective on the {W}igner-{A}raki-{Y}anase theorem.
\newblock {\em Journal of Physics: Conference Series}, 1638(1):012009, 2020.

\bibitem{LoveridgeBusch2011MeasurementofQMOperatorsRevisited}
L.~Loveridge and P.~Busch.
\newblock ``{M}easurement of quantum mechanical operators'' revisited.
\newblock {\em Eur. Phys. J. D}, 62(2):297--307, 2011.

\bibitem{mannalath_comprehensive_nodate}
V.~Mannalath, S.~Mishra, and A.~Pathak.
\newblock A comprehensive review of quantum random number generators: concepts,
  classification and the origin of randomness.
\newblock {\em Quantum Information Processing}, 22:439, 2023.

\bibitem{Naimark1940AboutSecondKindSelfAdjointExtensionsOfSymmetricalOperator}
M.~Naimark.
\newblock About second-kind self-adjoint extensions of symmetrical operator.
\newblock {\em Izv. Akad. Nauk USSR, Ser. Mat}, 4:53--104, 1940.

\bibitem{Ozawa1984QuantumMeasuringProcessesOfContinuousObservables}
M.~Ozawa.
\newblock Quantum measuring processes of continuous observables.
\newblock {\em J. Math. Phys.}, 25(1):79--87, 1984.

\bibitem{Spekkens2005ContextualityForPreparationsTransformationsAndUnsharpMeasurements}
R.~Spekkens.
\newblock Contextuality for preparations, transformations, and unsharp
  measurements.
\newblock {\em Phys. Rev. A}, 71(5):052108, 2005.

\bibitem{Stinespring1955PositiveFunctionsOnC-star-algebras}
W.~Stinespring.
\newblock Positive functions on {C*}-algebras.
\newblock {\em Proc. Amer. Math. Soc.}, 6(2):211--216, 1955.

\bibitem{Stoica2025CanWeAccuratelyReadWriteQuantumData}
O.~Stoica.
\newblock Can we accurately read or write quantum data?
\newblock {\em Physica Scripta}, 100(8):085108, 2025.

\bibitem{wigner1952MessungQMOperatoren}
E.~Wigner.
\newblock Die {M}essung quantenmechanischer {O}peratoren.
\newblock {\em {Zeitschrift f{\"u}r Physik A Hadrons and nuclei}},
  133(1):101--108, 1952.

\end{thebibliography}

\end{document}